\documentclass[11pt]{article}
\usepackage[utf8]{inputenc}
\usepackage{fullpage}
\usepackage{amsmath, amsthm, amssymb, mathtools}
\usepackage{hyperref}
\hypersetup{
    colorlinks=true,
    linkcolor=blue,
    filecolor=magenta,      
    urlcolor=cyan,
    citecolor=blue
}
\usepackage{tikz}

\newcommand{\ket}[1]{|#1\rangle}

\newcommand{\braket}[2]{\langle#1|#2\rangle}

\newtheorem{theorem}{Theorem}[section]

\newtheorem{lemma}[theorem]{Lemma}

\newtheorem{conjecture}[theorem]{Conjecture}

\newtheorem{definition}[theorem]{Definition}

\renewcommand{\epsilon}{\varepsilon}
\newcommand{\cbra}[1]{\left\{#1\right\}}
\newcommand{\rbra}[1]{\left(#1\right)}
\renewcommand{\sp}{\mathsf{span}}
\newcommand{\mathify}[1]{\ifmmode{#1}\else\mbox{$#1$}\fi}

\newcommand{\zone}{\{0, 1\}}

\newcommand{\ndeg}{\mathsf{ndeg}}
\renewcommand{\deg}{\mathsf{deg}}
\DeclarePairedDelimiter\abs{\lvert}{\rvert}%

\makeatletter
\let\oldabs\abs
\def\abs{\@ifstar{\oldabs}{\oldabs*}}
\makeatother

\newcommand{\AND}{\mathsf{AND}}
\newcommand{\EXACT}{\mathsf{EXACT}}
\newcommand{\OR}{\mathsf{OR}}
\renewcommand{\TH}{\mathsf{TH}}

\newcommand{\MOD}{\mathsf{MOD}}

\DeclareMathOperator{\QFT}{QFT}

\newcommand{\A}{\mathcal{A}}

\newcommand{\T}{\mathcal{T}}

\newcommand{\R}{\mathbb{R}}
\newcommand{\C}{\mathbb{C}}

\DeclareMathOperator*{\E}{\mathbb{E}}

\newcommand{\QE}{\mathsf{Q}_\mathsf{E}}
\newcommand{\QN}{\mathsf{Q}_\mathsf{N}}

\title{Exact quantum query complexity of computing\\
Hamming weight modulo powers of two and three}
\author{Arjan Cornelissen\thanks{Institute for Logic, Language, and Computation, University of Amsterdam and QuSoft. {\tt arjan@cwi.nl}}
\and
Nikhil S.~Mande\thanks{CWI, Amsterdam. Supported by the Dutch Research Council (NWO) through QuantERA ERA-NET Cofund project QuantAlgo (project number 680-91-034). {\tt Nikhil.Mande@cwi.nl}}
\and 
Maris Ozols\thanks{Institute for Logic, Language, and Computation, Korteweg-de Vries Institute for Mathematics, and
Institute for Theoretical Physics, University of Amsterdam and QuSoft. Supported by an NWO Vidi grant (Project No. VI.Vidi.192.109). {\tt marozols@gmail.com}}
\and 
Ronald de Wolf\thanks{QuSoft, CWI and University of Amsterdam, the Netherlands. Partially supported by the Dutch Research Council (NWO/OCW), as part of the Quantum Software Consortium programme (project number 024.003.037), and through QuantERA ERA-NET Cofund project QuantAlgo (680-91-034). {\tt rdewolf@cwi.nl}}
}
\date{}

\begin{document}

\maketitle

\begin{abstract}
We study the problem of computing the Hamming weight of an $n$-bit string modulo $m$, for any positive integer $m \leq n$ whose only prime factors are 2 and 3. We show that the exact quantum query complexity of this problem is $\left\lceil n(1 - 1/m) \right\rceil$.
The upper bound is via an iterative query algorithm whose core components are:
\begin{itemize}
    \item the well-known 1-query quantum algorithm (essentially due to Deutsch) to compute the Hamming weight a 2-bit string mod~2 (i.e., the parity of the input bits), and 
    \item a new 2-query quantum algorithm to compute the Hamming weight of a 3-bit string mod~3.
\end{itemize}
We show a matching lower bound (in fact for arbitrary moduli~$m$) via a variant of the polynomial method [de~Wolf, SIAM J.~Comput., 32(3), 2003]. This bound is for the weaker task 
of deciding whether or not a given $n$-bit input has Hamming weight 0 modulo $m$, and it holds even in the stronger \emph{non-deterministic} quantum query model where an algorithm must have positive acceptance probability iff its input evaluates to~1.
For $m>2$ our lower bound exceeds $n/2$, beating the best lower bound provable using the general polynomial method [Theorem 4.3, Beals et al., J.~ACM 48(4), 2001].
\end{abstract}


\section{Introduction}

Query complexity considers the number of queries to input variables needed to compute some function of the input. \emph{Quantum} query complexity has been well studied over the last few decades, and has been the source of many new quantum algorithms~\cite{deutsch&jozsa,simon:power,shor:factoring,grover:search,ambainis:edj,mnrs:searchwalk}.
An important special case is where the function we want to compute of input~$x\in\zone^n$ only depends on the Hamming weight~$|x|$, i.e., the number of 1s in~$x$. Such functions are called \emph{symmetric}, and well-studied examples of symmetric Boolean-valued functions include AND, OR, Majority, Parity, etc. The quantum query complexity (allowing small error probability) of all such functions was tightly characterized in~\cite{wolf:degreesymmf}.
Another important class of symmetric problems is \emph{approximate counting}, where we try to approximate the Hamming weight~$|x|$ itself rather than computing some Boolean property of~$|x|$. Note that this essentially corresponds to computing the most significant bits of $|x|$. For example, if $n$ is a power of~2, then approximating $|x|$ up to additive error $n/1024$ corresponds to computing the 10 most significant bits of~$|x|$.

We could also consider computing the \emph{least} significant bits of~$|x|$. For example, the $k$ least significant bits of $|x|$ correspond to $|x|$ modulo $2^k$ ($k=1$ would be Parity).
More generally, in this paper we study the quantum query complexity of computing the Hamming weight of an $n$-bit string modulo a fixed integer~$m$, given quantum query access to its bits. We consider the \emph{exact} setting where an algorithm is required to output the correct answer with certainty. It is easy to see by an adversarial argument that a classical exact (i.e., deterministic) algorithm needs to make $n$ queries, but the quantum query complexity of computing $|x|$ modulo~$m$ is much more subtle.

We give a quantum algorithm with query complexity $\left\lceil n(1 - 1/m) \right\rceil$ for all $m$ whose only prime factors are 2 and 3, and prove a perfectly matching lower bound.
In particular, our lower bound holds for all integers $m \geq 2$ (not just the ones with prime factors 2 and 3), and against the Boolean function $f: \zone^n \to \zone$ defined by $f(x) = 1$ iff $|x| \equiv 0$ modulo $m$ (clearly, $|x|$ modulo $m$ is at least as hard to compute as this $f$). Moreover, our lower bound holds in the weaker non-deterministic model~\cite{wolf:nqj}, where an algorithm is only required to output 1 with strictly positive probability iff $x$ is a 1-input. 

These functions against which our lower bound holds are examples of total symmetric Boolean functions, that is, they are defined on all inputs in $\zone^n$, and the function value only depends on the Hamming weight of the input. The exact quantum query complexity has been tightly characterized only for very few total symmetric functions. Beals et al.~\cite{bbcmw:polynomialsj} introduced the polynomial method to prove quantum query complexity lower bounds. Using this they showed that the exact quantum query complexity of the Parity function on $n$ variables equals $\left\lceil n/2 \right\rceil$ (the same lower bound for Parity was obtained independently around the same time in~\cite{fggs:parity}), and that the exact complexity of $\OR_n$ and $\AND_n$ equals $n$. They also showed using a result of von zur Gathen and Roche~\cite{GR97} that the exact quantum query complexity of any non-constant symmetric Boolean function on $n$ variables is at least $n/2 - O(n^{0.548})$. Ambainis, Iraids and Smotrovs~\cite{AIS13} showed a tight bound on the exact quantum query complexity of $\EXACT_k$ (the Boolean function that outputs 1 iff the Hamming weight of the input equals $k$) for all $k$, and $\TH_k$ (the Boolean threshold function that outputs 1 iff the Hamming weight of the input is at least $k$) for all $k$. Ambainis, Gruska and Zheng~\cite{AGZ15} showed that a function on $n$ input variables has exact quantum query complexity $n$ iff it is equal to $\AND_n$ up to negations and permutations of the input variables and negation of the output. More recently, Ambainis, Iraids and Nagaj~\cite{AIN17} showed a tight bound for the exact quantum query complexity of $\EXACT_{k,\ell}$ (the Boolean function that outputs 1 iff the Hamming weight of the input equals $k$ or $\ell$) for all $k, \ell$.

\subsection{Our results}\label{sec: our results}

The following is our main upper bound.
\begin{theorem}\label{thm: intro main upper bound}
Let $n$ be a positive integer and let $m > 1$ be an integer such that $m = 2^i3^j$ for some integers $i, j$. Then there is an exact quantum query algorithm that determines the Hamming weight of an $n$-bit string modulo $m$ by querying it at most $\left\lceil n\rbra{1 - \frac{1}{m}}\right\rceil$ times.
\end{theorem}
The two core components of our algorithm are the following two basic algorithms:
\begin{itemize}
    \item A 1-query algorithm that outputs the parity of 2 input bits. This algorithm is essentially due to Deutsch~\cite{deutsch:uqc}.
    \item A 2-query algorithm that outputs the Hamming weight of a 3-bit input modulo 3. We give this algorithm in Section~\ref{sec: upper bound}.
\end{itemize}
The overall algorithm for Theorem~\ref{thm: intro main upper bound} is iterative.
In the simplest case when $m = 2$ or $m = 3$, we partition the $n$ input bits into blocks of size $m$ and run the corresponding basic algorithm from above on each of the blocks, using $m-1$ queries per block. If the input size is not a multiple of $m$, we query each of the remaining bits individually. The mod-$m$ Hamming weight of the initial input is then the sum of all the outputs modulo $m$. It is not hard to show that the query complexity of this algorithm is at most $\left\lceil \frac{n}{m}\rbra{m-1}\right\rceil = \left\lceil n\rbra{1 - \frac{1}{m}}\right\rceil$.

Let us sketch our algorithm for the general case when $m = 2^i 3^j$ with $i + j > 1$. We write $m = m_1 m_2$, for some $m_1, m_2 < m$, and proceed recursively (see Figure~\ref{fig:alg}).
First, we partition the $n$-bit string into blocks of size $m_1$ and recursively determine the Hamming weight of each block modulo~$m_1$. This takes at most $\left\lceil n\rbra{1 - 1/m_1}\right\rceil$ queries in total. For the constant blocks where the string is either $0^{m_1}$ or $1^{m_1}$ the modulo-$m_1$ Hamming weight is $0$, while for the non-constant blocks we get the actual Hamming weight of each block.
There are fewer than $m_1$ leftover bits that did not fit in a full block, so we query each of them individually. Let $w$ denote the total Hamming weight of these bits and all non-constant blocks together.
It remains to determine the total Hamming weight of all constant blocks with respect to the composite modulus $m = m_1 m_2$.
We replace each constant block by a single bit of the same value and recursively compute the Hamming weight $w'$ of the resulting $n'$-bit string modulo $m_2$ using $\left\lceil n'\rbra{1 - 1/m_2}\right\rceil$ additional queries.
Since each variable of the shorter string corresponds to $m_1$ equal bits of the original string, the overall Hamming weight is $m_1 w'+ w$ modulo $m$. Careful analysis shows that the cost of this algorithm is at most $\left\lceil n\rbra{1 - \frac{1}{m}}\right\rceil$, proving Theorem~\ref{thm: intro main upper bound}.
For more details, see Section~\ref{sec: upper bound}.

\newcommand{\curlybrace}[3]{
\draw (#1-#3+0.7pt,#2-\r)
  arc[start angle = 180, end angle =  90] -- (#1-\r,#2) arc[start angle = -90, end angle =   0]
  arc[start angle = 180, end angle = 270] -- (#1+#3-\r-0.7pt,#2)
  arc[start angle =  90, end angle =   0];
}

\begin{figure}[t!]
\centering
\def\r{0.2cm} 
\def\w{0.5cm} 
\def\W{1.4cm} 
\def\H{14pt}  
\def\T{1.2cm} 
\def\D{2.0cm} 
\begin{tikzpicture}[thin, > = latex, radius = \r,
    block/.style = {draw, inner sep = 0, minimum height = \H, align = center, anchor = west},
    B/.style = {block, text width = \W},
    b/.style = {block, text width = \w}
  ]
  
  \node[align = center] at (3*\W,1.2*\T) {Constant\\blocks};
  \curlybrace{3*\W}{0.5*\T}{3*\W}
  \node at (4.5*\W,0) {$\cdots$};
  \node at (3*\W+1.5*\w,-\D) {...};
  \foreach \i/\a in {0/0, 1/1, 2/1, 3/0, 5/1} {
    \node[B] at (\i*\W,0) {\;$\a^{m_1}$};
    \node[b] at (3*\W-3*\w+\i*\w,-\D) {\a};
    \draw[->] (\i*\W+0.5*\W,-0.5*\H) -- (3*\W-2.5*\w+\i*\w,-\D+0.5*\H);
  }

  \node[align = center] at (8*\W,1.2*\T) {Non-constant\\blocks};
  \curlybrace{8*\W}{0.5*\T}{2*\W}
  \node at (8.5*\W,0) {$\cdots$};
  \node[B] at (6*\W,0) {$101101$};
  \node[B] at (7*\W,0) {$001010$};
  \node[B] at (9*\W,0) {$011011$};

  \def\more{0.9cm}
  \node[align = center] at (10*\W+0.5*\more,1.2*\T) {Leftover\\bits};
  \curlybrace{10*\W+0.5*\more}{0.5*\T}{0.5*\more}
  \node[B, text width = \more] at (10*\W,0) {$1011$};

  \begin{scope}[yscale = -1]
    \curlybrace{8*\W+0.5*\more}{0.5*\T}{3.25cm}
  \end{scope}
  \node[align = center] at (8*\W+0.5*\more,-1.2*\T) {For this string, we know the\\exact Hamming weight $w$.};

  \begin{scope}[yscale = -1]
    \curlybrace{3*\W}{\D+0.5*\T}{3*\w}
  \end{scope}
  \node[align = center] at (3*\W,-\D-1.2*\T) {Recurse to determine the\\Hamming weight $w'$ (mod $m_2$).};

\end{tikzpicture}
\caption{\label{fig:alg}Sketch of our algorithm for computing the Hamming weight modulo $m = m_1 m_2$.
First, we split the string into blocks of size $m_1$ and recursively compute their Hamming weight modulo~$m_1$.
For those blocks where the string is constant, the actual Hamming weight is either $0$ or $m_1$, while for non-constant blocks we know it exactly.
(Unlike depicted above, the two types of blocks may come in any order.)
Next, we individually query each of the leftover bits and denote by $w$ the total Hamming weight of these bits and all non-constant blocks together.
Finally, we shorten each constant block to a single bit and recursively determine the Hamming weight $w'$ modulo $m_2$ of the resulting string.
The Hamming weight modulo $m_1 m_2$ of the original string is then $m_1 w' + w$.}
\end{figure}
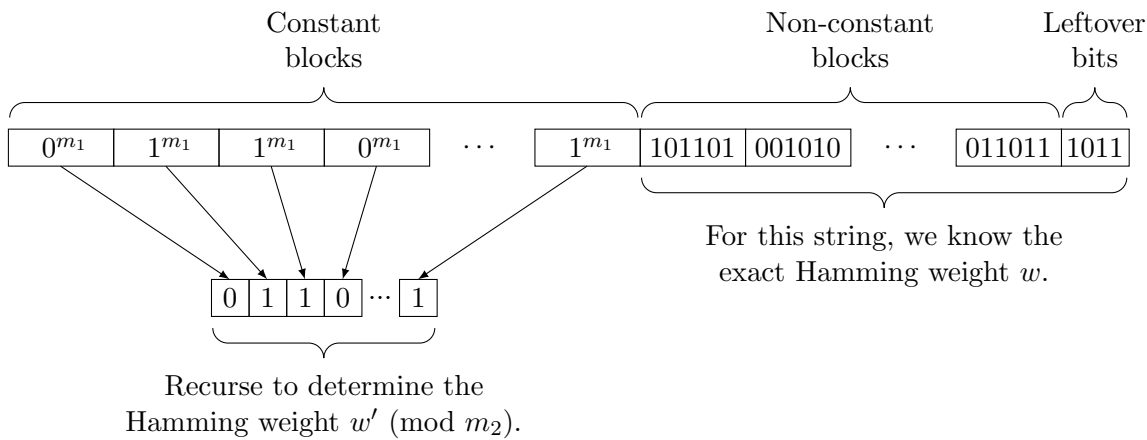

Our lower bound holds for more general Boolean functions, and also against non-deterministic quantum query algorithms~\cite{wolf:nqj}. A \emph{non-deterministic} query algorithm $\A$ for a Boolean function $f : \zone^n \to \zone$ has the constraint that $\Pr[\A(x) = 0] = 1$ for $x \in f^{-1}(0)$ and $\Pr[\A(x) = 1] > 0$ for $x \in f^{-1}(1)$. The \emph{cost} of $\A$ is the number of queries it makes (on the worst-case input), and the non-deterministic quantum query complexity of $f$, which we denote by $\QN(f)$, is the minimum cost of a non-deterministic quantum query algorithm for $f$. In contrast, the \emph{exact} query complexity of a Boolean function $f : \zone^n \to \zone$, which we denote by $\QE(f)$, is the minimum number of queries required by a quantum query algorithm that outputs $f(x)$ with probability 1, for all $x \in \zone^n$.

For integers $1 < m \leq n$, define the function $\MOD_m : \zone^n \to \zone$ by
\[
\MOD_m(x) = \begin{cases}
1 & |x| \equiv 0 \pmod m,\\
0 & \textnormal{otherwise}.
\end{cases}
\]
The following is our main lower bound.
\begin{theorem}\label{thm: intro main lower bound}
Let $1 < m \leq n$ be integers. Then,
\[
\QN(\MOD_m) \geq \left\lceil n\rbra{1 - \frac{1}{m}}\right\rceil.
\]
\end{theorem}

Our lower bound is proved via the non-deterministic variant of the polynomial method~\cite{wolf:nqj}. Since it exceeds $n/2$ whenever $m>2$, we cannot obtain this tight result from the standard polynomial method~\cite{bbcmw:polynomialsj} (which yields lower bounds that are at most $n/2$). 

Clearly the task of computing the Hamming weight of an input string $x$ modulo $m$ is at least as hard as computing $\MOD_m(x)$. Along with the simple observation that $\QE(f) \geq \QN(f)$ for all Boolean functions $f$, Theorems~\ref{thm: intro main upper bound} and~\ref{thm: intro main lower bound} yield our main result, stated below.

\begin{theorem}\label{thm: intro main}
Let $1 < m \leq n$ be integers such that $m = 2^i3^j$ for some integers $i, j$. Then the exact quantum query complexity of computing the Hamming weight of an $n$-bit string modulo $m$ is $\left\lceil n\rbra{1 - \frac{1}{m}}\right\rceil$.
\end{theorem}

Theorem~\ref{thm: intro main lower bound} implies that the lower bound in Theorem~\ref{thm: intro main} in fact holds for \emph{all} integers $m \leq n$. We conjecture that this is tight (see discussion in Section~\ref{sec:discussion}).

\begin{conjecture}\label{conj:main conjecture}
Let $1 < m \leq n$ be integers. Then the exact quantum query complexity of computing the Hamming weight of an $n$-bit string modulo $m$ is $\left\lceil n\rbra{1 - \frac{1}{m}}\right\rceil$.
\end{conjecture}

\section{Preliminaries}
For a string $x \in \zone^n$, we index its coordinates by the set $[n] = \cbra{1,2, \dots, n}$, that is, $x = x_1x_2\ldots x_n$. Let $|x| = |\cbra{i \in [n] : x_i = 1}|$ denote the Hamming weight of $x$. For a set $S \subseteq [n]$, let $x_S$ denote the restriction of $x$ to the indices in $S$. That is, $x_S$ is the string in $\zone^S$ defined by $(x_S)_i = x_i$ for all $i \in S$.

An $n$-variate multilinear polynomial $p$ is a function $p : \R^n \to \C$ that can be expressed as
\[
p(x) = \sum_{S \subseteq [n]} a_S \prod_{i \in S} x_i,
\]
for some $a_S \in \C$ for all $S \subseteq [n]$. The \emph{degree} of $p$ is defined as $\deg(p) := \max_{S \subseteq [n]}\cbra{|S| : a_S \neq 0}$. It is well known that every total Boolean function $f : \zone^n \to \zone$ has a unique multilinear polynomial $p: \R^n \to \C$ such that $f(x) = p(x)$ for all $x \in \zone^n$. Let $\deg(f)$ denote the degree of this polynomial.

We refer the reader to~\cite{nielsen&chuang:qc,LECTURENOTES} for the basics of quantum computing. A quantum algorithm acts on some finite-dimensional Hilbert space. It starts in some fixed initial state and alters it through unitary operations and measurements. A quantum \emph{query} algorithm is only allowed to access the input through one particular unitary operation which is referred to as the ``oracle'', i.e., the initial state, measurement operations, and all the other unitary operations must be independent of the input. In the standard query model, we access a bit string $x \in \{0,1\}^n$ by applying an oracle~$O_x$ given by the diagonal matrix
\[
  O_x =
  \begin{bmatrix}
    (-1)^{x_1} & \\
    & (-1)^{x_2} & \\
    && \ddots & \\
    &&& (-1)^{x_n}
  \end{bmatrix}.
\]
There may be additional dimensions where the oracle $O_x$ acts trivially
(which is needed to be able to apply a controlled version of $O_x$, otherwise an algorithm could not see the difference between $x$ and its complement).

\begin{definition}[Non-deterministic degree]\label{defn: ndeg}
A non-deterministic polynomial for a Boolean function $f: \zone^n \to \zone$ is a multilinear polynomial $p : \R^n \to \C$ such that for all $x\in\zone^n$, $p(x) = 0$ iff $f(x) = 0$ (i.e., $f$ and $p$ have the same support on the Boolean cube).
The non-deterministic degree of~$f$, denoted $\ndeg(f)$, is the minimum degree among all such~$p$.
\end{definition}

\cite[Theorem 2.3]{wolf:nqj} showed via a probabilistic argument that the non-deterministic quantum query complexity (see Section~\ref{sec: our results}) of a Boolean function $f : \zone^n \to \zone$ is bounded below by $\ndeg(f)$.\footnote{In fact that theorem showed $\QN(f)=\ndeg(f)$ for all $f$, but the upper bound does not concern us here.} This contrasts with the usual polynomial method lower bound $\QE(f)\geq \deg(f)/2$ where the factor-2 is sometimes necessary, as witnessed by the Parity function~\cite{bbcmw:polynomialsj,fggs:parity}.

\begin{theorem}[{\cite{wolf:nqj}}]\label{thm: qn ndeg}
For all Boolean functions $f : \zone^n \to \zone$, we have $\QN(f) \geq \ndeg(f)$.
\end{theorem}

We require the following result regarding symmetrization of multivariate polynomials, due to Minsky and Papert~\cite{minsky&papert:perceptrons}. For completeness we give a proof below with complex ranges.

\begin{theorem}[\cite{minsky&papert:perceptrons}]\label{thm: mp}
Let $p : \R^n \to \C$ be a multilinear polynomial. Then there exists a univariate polynomial $q : \R \to \C$ such that $\deg(q) \leq \deg(p)$ and for all $k \in \cbra{0,1,\dots,n}$,
\[
q(k) = \E_{x \in \zone^n : |x| = k} [p(x)].
\]
\end{theorem}

\begin{proof}
Define a symmetrized $n$-variate polynomial $\overline{p}(x) = \frac{1}{n!} \sum_{\pi \in \mathrm{S}_n} p(\pi(x))$ where $\mathrm{S}_n$ is the symmetric group and $\pi(x)$ denotes the variables $x$ permuted according to the permutation $\pi$.
The degree~$d$ of $\overline{p}$ is at most $\deg(p)$.
Since $p$ is multilinear, the averaging over $\mathrm{S}_n$ causes all monomials of degree $k$ in $\overline{p}$ to have the same coefficient $c_k \in \C$:
\[
\overline{p}(x) = \sum_{k=0}^{d} c_k \sum_{S\subseteq[n]:|S|=k}\prod_{i\in S} x_i.
\]
When restricting to inputs of the form $x \in \zone^n$, we have
\[
  \sum_{S\subseteq[n]:|S|=k}\prod_{i\in S} x_i = \binom{|x|}{k},
\]
so we can rewrite $\overline{p}(x)$ as follows:
\[
\overline{p}(x)
=\sum_{k=0}^{d} c_k \binom{|x|}{k}
=\sum_{k=0}^{d} \frac{c_k}{k!}|x|(|x|-1)\cdots (|x|-k+1).
\]
The latter can be viewed as a univariate polynomial $q$ in $|x|$ of degree $d\leq \deg(p)$.
To show that $q(k)$ agrees with the expectation of $p$ over inputs of Hamming weight $k$, let $y \in \zone^n$ be any string such that $|y| = k$. Then
\[
\frac{1}{\binom{n}{k}} \sum_{x \in \zone^n : |x| = k} p(x)
= \frac{1}{\binom{n}{k}} \frac{1}{k!(n-k)!} \sum_{\pi \in \mathrm{S}_n} p(\pi(y))
= \frac{1}{n!} \sum_{\pi \in \mathrm{S}_n} p(\pi(y))
= \overline{p}(y)
= q(k),
\]
as desired.
\end{proof}

\section{Upper bound}\label{sec: upper bound}

In this section, we prove Theorem~\ref{thm: intro main upper bound}. We exhibit an iterative algorithm to compute the modulo-$m$ Hamming weight of an input bit-string $x$, for any positive integer $m$ whose only prime factors are 2 and 3. Theorem~\ref{thm: intro main upper bound} immediately follows from the following statement.

\begin{theorem}\label{thm: main upper bound}
Let $n$ be a positive integer and let $m > 1$ be an integer such that $m = 2^i3^j$ for some integers $i, j$. Then there is an exact quantum query algorithm that queries $x\in\zone^n$ at most $\left\lceil n\rbra{1 - \frac{1}{m}}\right\rceil$ times and outputs the following:
\begin{itemize}
    \item A partition $S_1 \cup S_2 = [n]$ such that $|S_1|$ is a multiple of $m$,
    \item the Hamming weight of $x_{S_2}$,
    \item a further partition of $S_1$ into blocks of size $m$, with the property that for each block, $x$ restricted to the indices in that block is constant (i.e., either $0^m$ or $1^m$).
\end{itemize}
\end{theorem}
First note that the algorithm described above immediately implies that the Hamming weight of an input $x$ modulo $m$ equals $|x_{S_2}|$ modulo $m$, since $|x_{S_1}| \equiv 0$ modulo $m$. Hence, Theorem~\ref{thm: main upper bound} implies Theorem~\ref{thm: intro main upper bound}.
We require the following two subroutines, the first of which is known as Deutsch's algorithm \cite{deutsch:uqc,cemm:revisited}. We include its proof for completeness.

\begin{lemma}[Deutsch's algorithm]\label{lemma: mod2 2 bits}
There exists an exact 1-query quantum algorithm that determines the Hamming weight modulo 2 of a 2-bit string.
\end{lemma}

\begin{proof}
Let
\[
H = 
\frac{1}{\sqrt{2}} \left[
\begin{array}{rr}
1 & 1\\ 
1 & -1
\end{array}\right],
\qquad O_x = \begin{bmatrix}
(-1)^{x_1} & 0\\
0 & (-1)^{x_2}
\end{bmatrix}.
\]
Consider the 1-query algorithm that prepares $H O_x H \ket{0}$, followed by a measurement in the computational basis.
It is easy to see that this outputs the parity of $x_1$ and $x_2$.
\end{proof}

\begin{lemma}\label{lemma: mod3 3 bits}
There exists an exact 2-query quantum algorithm that determines the Hamming weight modulo 3 of a 3-bit string.
\end{lemma}
To the best of our knowledge, such an algorithm was not known prior to our work.

\begin{proof}
The algorithm will use a $5$-dimensional state space. The oracle $O_x$ accessing the input string $x \in \zone^3$ is given by
\[O_x = \begin{bmatrix}
    (-1)^{x_1} & 0 & 0 & 0 & 0 \\
    0 & (-1)^{x_2} & 0 & 0 & 0 \\
    0 & 0 & (-1)^{x_3} & 0 & 0 \\
    0 & 0 & 0 & 1 & 0 \\
    0 & 0 & 0 & 0 & 1
\end{bmatrix}.\]
Furthermore, let $\omega = e^{\frac{2\pi i}{3}}$ and
\begin{align*}
    \QFT &= \frac{1}{\sqrt{3}}\begin{bmatrix}
    	1 & 1 & 1 & 0 & 0 \\
    	1 & \omega & \omega^2 & 0 & 0 \\
    	1 & \omega^2 & \omega & 0 & 0 \\
    	0 & 0 & 0 & 1 & 0 \\
    	0 & 0 & 0 & 0 & 1
    \end{bmatrix}
\end{align*}
denote the $3 \times 3$ quantum Fourier transform acting only on the first three basis states.
The two non-query operations are given by
\begin{align*}
    U &= \frac14\begin{bmatrix}
		4 & 0 & 0 & 0 & 0 \\
		0 & -\frac{1}{2} + i\frac{3\sqrt{3}}{2} & 0 & 3 & 0 \\
		0 & 0 & -\frac{1}{2} - i\frac{3\sqrt{3}}{2} & 0 & 3 \\
		0 & 3 & 0 & \frac{1}{2} + i\frac{3\sqrt{3}}{2} & 0 \\
		0 & 0 & 3 & 0 & \frac{1}{2} - i\frac{3\sqrt{3}}{2}
	\end{bmatrix}
\end{align*}
and
\begin{align*}
    V &= \frac{1}{\sqrt{2}}\begin{bmatrix}
		\sqrt{2} & 0 & 0 & 0 & 0 \\
		0 & \frac{1}{2} - i\frac{\sqrt{3}}{2} & 0 & 1 & 0 \\
		0 & 0 & \frac{1}{2} + i\frac{\sqrt{3}}{2} & 0 & 1 \\
		0 & 1 & 0 & - \frac{1}{2} - i\frac{\sqrt{3}}{2} & 0 \\
		0 & 0 & 1 & 0 & - \frac{1}{2} + i\frac{\sqrt{3}}{2}
	\end{bmatrix}.
\end{align*}
One may check by inspection that $U$ and $V$ are unitary.
We denote the computational basis of the underlying 5-dimensional space by $\{\ket{0}, \ket{1}, \dotsc, \ket{4}\}$, and define orthogonal subspaces 
\begin{align*}
    S_0 & = \sp\cbra{\ket{0}},\\
    S_1 & = \sp\cbra{\ket{1},\ket{2}},\\
    S_2 & = \sp\cbra{\ket{3},\ket{4}},
\end{align*}
corresponding to values of the Hamming weight modulo 3.
Let $\Pi_i$ denote the projector onto $S_i$, for each $i \in \cbra{0, 1, 2}$.
Our algorithm prepares the state
\begin{equation}
\ket{\psi(x)} =
V (\QFT^{\dagger}O_x\QFT) U (\QFT^{\dagger}O_x\QFT) \ket{0},
\label{eq:psi}
\end{equation}
and outputs the result obtained by measuring this state with respect to the projectors $\{\Pi_0, \Pi_1, \Pi_2\}$. The query complexity of this algorithm is 2.
For convenience, we list in Appendix~\ref{app: intermediate states} the intermediate states of the algorithm for all inputs $x \in \zone^3$.
Observe from the table in Appendix~\ref{app: intermediate states} that $\ket{\psi(x)} \in S_i$ iff $|x| \equiv i \pmod 3$, showing the correctness of our algorithm.
\end{proof}

\begin{proof}[Proof of Theorem~\ref{thm: main upper bound}]
We prove this by induction on~$m$.

\textbf{Base case:} First, suppose that $m = 2$. We divide $[n]$ into pairs, run the algorithm from Lemma~\ref{lemma: mod2 2 bits} on each of the pairs separately, and simply query the one remaining bit if $n$ is odd. Since $x$ is constant on all pairs that give outcome $0$, we group them into $S_1$ and remark that $|x_{S_1}|$ indeed is a multiple of $2$. Furthermore, we let $S_2 = [n] \setminus S_1$, and observe that we can now easily calculate $|x_{S_2}|$ from the measurement outcomes. The number of queries used is indeed $\lceil n/2\rceil = \lceil n(1-1/m)\rceil$. This completes the base case for $m = 2$. The base case for $m=3$ follows analogously using Lemma~\ref{lemma: mod3 3 bits}.

\textbf{Inductive step:} Let $m \geq 4$ and suppose that the statement holds for all $m' < m$. The modulus $m = 2^i 3^j$ with $i+j > 1$ cannot be prime, so we can factor it as $m = m_1m_2$ such that $m_1 < m$ and $m_2 < m$. We perform the following procedure (see Figure~\ref{fig:alg}):
\begin{enumerate}
    \item Run the algorithm from our induction hypothesis with mod $m_1$, making at most $\lceil n(1-1/m_1) \rceil$ queries. This returns a partition $S_1' \cup S_2' = [n]$, a further partition of $S_1' = B_1 \cup \cdots \cup B_{\ell}$ into blocks of size $m_1$ such that $x$ is constant on each of these blocks, and an integer $w$ such that $|x_{S_2'}| = w$. Note that $\ell \leq \lfloor n/m_1\rfloor$.
    \item Since $x$ is constant on each $B_j$, we choose an arbitrary element $a_j \in B_j$, for all $j \in [\ell]$, and let $A = \{a_j : j \in [\ell]\} \subseteq [n]$. Run the algorithm from our induction hypothesis on $x_A$, with mod $m_2$, costing at most $\lceil \ell(1-1/m_2) \rceil$ queries. This returns a partition $A_1' \cup A_2' = A$ and a further partition of $A_1' = B_1' \cup \cdots \cup B_{\ell'}'$ into blocks of size $m_2$ such that $x$ is constant on each of these blocks, and an integer $w'$ such that $|x_{A_2'}| = w'$.
    \item For $i \in [\ell']$, let $L_i := \cbra{j \in [\ell] : a_j \in B'_{i}}$. By the construction from the previous step, $|L_i| = m_2$ for all $i \in [\ell']$. For every block $B_i' = \{a_j : j \in L_i\}$, define $C_i = \cup_{j \in L_i} B_j$ (i.e., $C_i$ represents all indices in $[n]$ that ``contribute'' to $B'_{i}$). Since $x$ is constant on each $B_i'$ and on each $B_j$, and since each element of $B_i'$ represents an entire $B_j$, we conclude that $x$ must also be constant on the whole $C_i$. Moreover, since $C_i$ is a disjoint union of $|L_i| = m_2$ sets, each of size $|B_j| = m_1$, we have $|C_i| = m_1m_2 = m$. Group all these $C_i$'s into a set~$S_1$, and notice that $|x_{S_1}|\equiv 0$ mod~$m$. Let $S_2 = [n] \setminus S_1$ denote the rest. If we denote $M := \cbra{j \in [\ell] : a_j \in A'_2}$, then we can write $S_2$ as $S_2' \cup (\cup_{j \in M} B_j)$, and hence
    \[|x_{S_2}| = |x_{S_2'}| + \sum_{j \in M} |x_{B_j}| = w + \sum_{a_j \in A_2'} |x_{a_j}| \cdot m_1 = w + w' m_1.\]
\end{enumerate}
Since $\ell \leq \lfloor n/m_1\rfloor$, the total number of queries of the entire algorithm is at most
\begin{align*}
    \left\lceil n\left(1-\frac{1}{m_1}\right)\right\rceil + \left\lceil \left\lfloor \frac{n}{m_1}\right\rfloor \left(1-\frac{1}{m_2}\right)\right\rceil &= n - \left\lfloor \frac{n}{m_1}\right\rfloor + \left\lfloor \frac{n}{m_1}\right\rfloor - \left\lfloor \left\lfloor \frac{n}{m_1} \right\rfloor \frac{1}{m_2} \right\rfloor = n - \left\lfloor \frac{n}{m_1m_2}\right\rfloor \\
    &= \left\lceil n - \frac{n}{m_1m_2}\right\rceil = \left\lceil n\left(1 - \frac{1}{m_1m_2}\right)\right\rceil,
\end{align*}
where we used that for every three integers $a,b,c > 0$, it holds that $\lfloor \lfloor a/b \rfloor /c \rfloor = \lfloor a/(bc) \rfloor$. Since $m=m_1m_2$, this completes the inductive step.
\end{proof}

\section{Lower bound}

In this section, we show that the upper bound in Theorem~\ref{thm: intro main upper bound} is tight. We do this by showing Theorem~\ref{thm: intro main lower bound}, which is a stronger lower bound: the exact query complexity of determining whether the Hamming weight of an $n$-bit string is 0 modulo $m$ or not, is at least $\left\lceil n\rbra{1 - 1/m} \right\rceil$. Moreover, our bound works in the stronger \emph{non-deterministic} model. Along with Theorem~\ref{thm: intro main upper bound}, this implies the lower bound in Theorem~\ref{thm: intro main} since computing the Hamming weight of an $n$-bit string modulo $m$ is at least as hard as deciding whether its Hamming weight is 0 modulo $m$. We first show a more general statement.

\begin{theorem}\label{thm: lower bound numberof0s}
Let $f: \zone^n \to \zone$ be a symmetric Boolean function such that $f(0^n) = 1$. Then,
\[
\QN(f) \geq \abs{\cbra{i \in [n] : f(x) = 0~\textnormal{if}~|x| = i}}.
\]
\end{theorem}

\begin{proof}
Let $p : \R^n \to \C$ be a non-deterministic polynomial for $f$. Theorem~\ref{thm: mp} implies existence of a polynomial $q : \R \to \C$ with $\deg(q) \leq \deg(p)$ that satisfies the following properties:
\begin{itemize}
    \item $q(0) = p(0^n) \neq 0$, since $f(0^n) = 1$,
    \item if $p(x) = 0$ and $|x| = i \in [n]$ then
    \[
    q(i) = \E_{x \in \zone^n : |x| = i} [p(x)] = 0.
    \]
\end{itemize}
Thus $q$ is a univariate non-constant polynomial that has at least as many roots as the number of Hamming weights on which $f$ outputs 0. Since the degree of a non-zero polynomial is at least as large as its number of roots,
\[
\ndeg(f) \geq \deg(p) \geq \deg(q) \geq \abs{\cbra{i \in [n] : f(x) = 0~\textnormal{if}~|x| = i}},
\]
which completes the proof after applying Theorem~\ref{thm: qn ndeg}.
\end{proof}

Theorem~\ref{thm: intro main lower bound} immediately follows, because the  $\MOD_m$ function takes value 1 on exactly $1+\lfloor n/m\rfloor$ of the $n+1$ possible Hamming weights, and hence is 0 on the other $n+1-(1+\lfloor n/m\rfloor)=\left\lceil n\rbra{1 - \frac{1}{m}}\right\rceil$ Hamming weights.

\section{Discussion}\label{sec:discussion}

We showed how to recover the Hamming weight of a bit string $x \in \{0,1\}^n$, modulo some integer $m$, where $m$ only has prime factors $2$ and $3$. The core building blocks we used are the algorithms from Lemmas~\ref{lemma: mod2 2 bits}~and~\ref{lemma: mod3 3 bits}, which solve the problem in the case where $n = m = 2$ and $n = m = 3$, respectively. We could resolve Conjecture~\ref{conj:main conjecture} via the same recursive proof method as in Theorem~\ref{thm: main upper bound} if we had algorithms for $n = m = p$ for all primes~$p$. Hence, a natural follow-up question is whether we can construct a 4-query quantum algorithm to compute $|x|$ mod $5$ for all  $x\in\zone^5$. 

\bibliography{bibo}

\appendix

\section{Intermediate states of our modulo-$3$ algorithm}\label{app: intermediate states}

For a given input $x \in \zone^3$, denote the intermediate states of our algorithm of Equation~\eqref{eq:psi} by
\begin{align*}
  \ket{\psi_1(x)} &= \widetilde{O}_x \ket{0}, \\
  \ket{\psi_2(x)} &= U \ket{\psi_1(x)}, \\
  \ket{\psi_3(x)} &= \widetilde{O}_x \ket{\psi_2(x)}, \\
  \ket{\psi_4(x)} &= V\ket{\psi_3(x)},
\end{align*}
where $\widetilde{O}_x = \QFT^{\dagger} O_x \QFT$.
Explicit expressions of these states are provided in Table~\ref{table: states}.
Curiously, the Gram matrix $G_{x,y} = \braket{\psi_4(x)}{\psi_4(y)}$ of the final states has a particularly elegant form, which we mention here with a view towards possible generalization to prime moduli 5 and higher:
\[
G = \frac{1}{2}
\begin{bmatrix}
 2 & 0 & 0 & 0 & 0 & 0 & 0 & 2 \\
 0 & 2 & -1 & 0 & -1 & 0 & 0 & 0 \\
 0 & -1 & 2 & 0 & -1 & 0 & 0 & 0 \\
 0 & 0 & 0 & 2 & 0 & -1 & -1 & 0 \\
 0 & -1 & -1 & 0 & 2 & 0 & 0 & 0 \\
 0 & 0 & 0 & -1 & 0 & 2 & -1 & 0 \\
 0 & 0 & 0 & -1 & 0 & -1 & 2 & 0 \\
 2 & 0 & 0 & 0 & 0 & 0 & 0 & 2
\end{bmatrix}.
\]
Letting $a_i = (-1)^{x_i}$ and $b_i = (-1)^{y_i}$, by a direct calculation we can obtain the following explicit formula for the entries of $G$ as a function of $a,b \in \{-1,1\}^3$:
\begin{align*}
    48 G_{a,b}
    &= 6 (a_1 b_1 + a_2 b_2 + a_3 b_3)
    + (a_1^2 b_1^2 + a_2^2 b_2^2 + a_3^2 b_3^2)
    - 3 (a_1 b_2 + a_1 b_3 + a_2 b_1 + a_2 b_3 + a_3 b_1 + a_3 b_2) \\
    &\quad + 3 (a_1^2 b_1 b_2 + a_2^2 b_2 b_3 + a_3^2 b_1 b_3 + a_1 a_2 b_1^2 + a_2 a_3 b_2^2 + a_1 a_3 b_3^2) \\
    &\quad + 9 (a_1 a_2 b_1 b_2 + a_2 a_3 b_2 b_3 + a_1 a_3 b_1 b_3).
\end{align*}
Using the fact that $a_i^2 = b_i^2 = 1$, we can simplify it to
\begin{align*}
    16 G_{a,b} = 1
   &+   (a_1 a_2 + a_1 a_3 + a_2 a_3)
    +   (b_1 b_2 + b_1 b_3 + b_2 b_3)
    + 2 (a_1 b_1 + a_2 b_2 + a_3 b_3) \\
   &+ 3 (a_1 a_2 b_1 b_2 + a_1 a_3 b_1 b_3 + a_2 a_3 b_3 b_2)
    -   (a_1 b_2 + a_1 b_3 + a_2 b_1 + a_2 b_3 + a_3 b_1 + a_3 b_2).
\end{align*}

\begin{table}[h!]
\scalebox{0.75}{
	\begin{tabular}{r|cccccccc}
		State on input $x$ & $000$ & $100$ & $010$ & $001$ & $011$ & $101$ & $110$ & $111$ \\\hline
		
		$\ket{\psi_1(x)} = \widetilde{O}_x \ket{0}$ & $\begin{bmatrix}
		1 \\
		0 \\
		0 \\
		0 \\
		0
		\end{bmatrix}$ & $\begin{bmatrix}
		\frac13 \\
		-\frac23 \\
		-\frac23 \\
		0 \\
		0
		\end{bmatrix}$ & $\begin{bmatrix}
		\frac13 \\
		\frac13+\frac{\sqrt{3}}{3}i \\
		\frac13-\frac{\sqrt{3}}{3}i \\
		0 \\
		0
		\end{bmatrix}$ & $\begin{bmatrix}
		\frac13 \\
		\frac13-\frac{\sqrt{3}}{3}i \\
		\frac13+\frac{\sqrt{3}}{3}i \\
		0 \\
		0
		\end{bmatrix}$ & $\begin{bmatrix}
		-\frac13 \\
		\frac23 \\
		\frac23 \\
		0 \\
		0
		\end{bmatrix}$ & $\begin{bmatrix}
		-\frac13 \\
		-\frac13-\frac{\sqrt{3}}{3}i \\
		-\frac13+\frac{\sqrt{3}}{3}i \\
		0 \\
		0
		\end{bmatrix}$ & $\begin{bmatrix}
		-\frac13 \\
		-\frac13+\frac{\sqrt{3}}{3}i \\
		-\frac13-\frac{\sqrt{3}}{3}i \\
		0 \\
		0
		\end{bmatrix}$ & $\begin{bmatrix}
		-1 \\
		0 \\
		0 \\
		0 \\
		0
		\end{bmatrix}$ \\
		
		$\ket{\psi_2(x)} = U\ket{\psi_1(x)}$ & $\begin{bmatrix}
		1 \\
		0 \\
		0 \\
		0 \\
		0
		\end{bmatrix}$ & $\begin{bmatrix}
		\frac13 \\
		\frac{1}{12}-\frac{\sqrt{3}}{4}i \\
		\frac{1}{12}+\frac{\sqrt{3}}{4}i \\
		-\frac12 \\
		-\frac12
		\end{bmatrix}$ & $\begin{bmatrix}
		\frac13 \\
		-\frac{5}{12}+\frac{\sqrt{3}}{12}i \\
		-\frac{5}{12}-\frac{\sqrt{3}}{12}i \\
		\frac14+\frac{\sqrt{3}}{4}i \\
		\frac14-\frac{\sqrt{3}}{4}i
		\end{bmatrix}$ & $\begin{bmatrix}
		\frac13 \\
		\frac13+\frac{\sqrt{3}}{6}i \\
		\frac13-\frac{\sqrt{3}}{6}i \\
		\frac14-\frac{\sqrt{3}}{4}i \\
		\frac14+\frac{\sqrt{3}}{4}i
		\end{bmatrix}$ & $\begin{bmatrix}
		-\frac13 \\
		-\frac{1}{12}+\frac{\sqrt{3}}{4}i \\
		-\frac{1}{12}-\frac{\sqrt{3}}{4}i \\
		\frac12 \\
		\frac12
		\end{bmatrix}$ & $\begin{bmatrix}
		-\frac13 \\
		\frac{5}{12}-\frac{\sqrt{3}}{12}i \\
		\frac{5}{12}+\frac{\sqrt{3}}{12}i \\
		-\frac14-\frac{\sqrt{3}}{4}i \\
		-\frac14+\frac{\sqrt{3}}{4}i
		\end{bmatrix}$ & $\begin{bmatrix}
		-\frac13 \\
		-\frac13-\frac{\sqrt{3}}{6}i \\
		-\frac13+\frac{\sqrt{3}}{6}i \\
		-\frac14+\frac{\sqrt{3}}{4}i \\
		-\frac14-\frac{\sqrt{3}}{4}i
		\end{bmatrix}$ & $\begin{bmatrix}
		-1 \\
		0 \\
		0 \\
		0 \\
		0
		\end{bmatrix}$ \\
		
		$\ket{\psi_3(x)} = \widetilde{O}_x \ket{\psi_2(x)}$ & $\begin{bmatrix}
		1 \\
		0 \\
		0 \\
		0 \\
		0
		\end{bmatrix}$ & $\begin{bmatrix}
		0 \\
		-\frac14-\frac{\sqrt{3}}{4}i \\
		-\frac14+\frac{\sqrt{3}}{4}i \\
		-\frac12 \\
		-\frac12
		\end{bmatrix}$ & $\begin{bmatrix}
		0 \\
		-\frac14+\frac{\sqrt{3}}{4}i \\
		-\frac14-\frac{\sqrt{3}}{4}i \\
		\frac14+\frac{\sqrt{3}}{4}i \\
		\frac14-\frac{\sqrt{3}}{4}i
		\end{bmatrix}$ & $\begin{bmatrix}
		0 \\
		\frac12 \\
		\frac12 \\
		\frac14-\frac{\sqrt{3}}{4}i \\
		\frac14+\frac{\sqrt{3}}{4}i
		\end{bmatrix}$ & $\begin{bmatrix}
		0 \\
		-\frac14-\frac{\sqrt{3}}{4}i \\
		-\frac14+\frac{\sqrt{3}}{4}i \\
		\frac12 \\
		\frac12
		\end{bmatrix}$ & $\begin{bmatrix}
		0 \\
		-\frac14+\frac{\sqrt{3}}{4}i \\
		-\frac14-\frac{\sqrt{3}}{4}i \\
		-\frac14-\frac{\sqrt{3}}{4}i \\
		-\frac14+\frac{\sqrt{3}}{4}i
		\end{bmatrix}$ & $\begin{bmatrix}
		0 \\
		\frac12 \\
		\frac12 \\
		-\frac14+\frac{\sqrt{3}}{4}i \\
		-\frac14-\frac{\sqrt{3}}{4}i
		\end{bmatrix}$ & $\begin{bmatrix}
		1 \\
		0 \\
		0 \\
		0 \\
		0
		\end{bmatrix}$ \\
		
		$\ket{\psi_4(x)} = V\ket{\psi_3(x)}$ & $\begin{bmatrix}
		1 \\
		0 \\
		0 \\
		0 \\
		0
		\end{bmatrix}$ & $\begin{bmatrix}
		0 \\
		-\frac{\sqrt{2}}{2} \\
		-\frac{\sqrt{2}}{2} \\
		0 \\
		0
		\end{bmatrix}$ & $\begin{bmatrix}
		0 \\
		\frac{\sqrt{2}}{4}+\frac{\sqrt{6}}{4}i \\
		\frac{\sqrt{2}}{4}-\frac{\sqrt{6}}{4}i \\
		0 \\
		0
		\end{bmatrix}$ & $\begin{bmatrix}
		0 \\
		\frac{\sqrt{2}}{4}-\frac{\sqrt{6}}{4}i \\
		\frac{\sqrt{2}}{4}+\frac{\sqrt{6}}{4}i \\
		0 \\
		0
		\end{bmatrix}$ & $\begin{bmatrix}
		0 \\
		0 \\
		0 \\
		-\frac{\sqrt{2}}{4}-\frac{\sqrt{6}}{4}i \\
		-\frac{\sqrt{2}}{4}+\frac{\sqrt{6}}{4}i
		\end{bmatrix}$ & $\begin{bmatrix}
		0 \\
		0 \\
		0 \\
		-\frac{\sqrt{2}}{4}+\frac{\sqrt{6}}{4}i \\
		-\frac{\sqrt{2}}{4}-\frac{\sqrt{6}}{4}i
		\end{bmatrix}$ & $\begin{bmatrix}
		0 \\
		0 \\
		0 \\
		\frac{\sqrt{2}}{2} \\
		\frac{\sqrt{2}}{2}
		\end{bmatrix}$ & $\begin{bmatrix}
		1 \\
		0 \\
		0 \\
		0 \\
		0
		\end{bmatrix}$
	\end{tabular}}
    \caption{\label{table: states}Intermediate states of our algorithm in Lemma~\ref{lemma: mod3 3 bits} for computing the Hamming weight modulo $3$. Recall from Equation~\eqref{eq:psi} that the oracle is in the Fourier basis: $\widetilde{O}_x = \QFT^{\dagger} O_x \QFT$.}
\end{table}


\end{document}